\documentclass[11pt]{article}
\usepackage{bm,amsfonts,amssymb,amsthm,amsmath}
\usepackage{mathrsfs}
\usepackage{graphicx,color}
\usepackage{CJK}
\usepackage{indentfirst,enumerate}

\renewcommand{\d}{\mathrm{d}}
\renewcommand{\H}{\mathcal{H}}
\newcommand{\J}{\mathcal{J}}

\newtheorem{theorem}{Theorem}
\newtheorem{lemma}[theorem]{Lemma}

\definecolor{orange}{rgb}{1,0.5,0}
\definecolor{rb}{rgb}{1,0,1}
\allowdisplaybreaks[1]

\begin{document}

\title{\textbf{The transmission of symmetry in liquid crystals}}
\author{Jie Xu \& Pingwen Zhang\footnote{Corresponding author}\\[2mm]
{\small LMAM \& School of Mathematical Sciences, Peking University, 
  Beijing 100871, China}\\[1mm]
{\small E-mail: rxj\_2004@126.com,\, pzhang@pku.edu.cn}\\
}
\date{\today}
\maketitle

\begin{abstract}
The existing experiments and simulations suggest that the molecular symmetry 
is always transmitted to homogeneous phases in liquid crystals. 
It has been proved for rod-like molecules. 
We conjecture that it holds for three other symmetries, 
and prove it for some molecules of these symmetries. 
\end{abstract}

\section{Introduction}
The application of liquid crystals benefits from their subtlety in anistropy, 
which origins from the anistropy at the moleculer level. 
Let us consider a rod-like molecule. Except for its location $\bm{x}$, 
we need to express its orientation by a unit vector $\bm{m}$. 
The distribution $f$ thus depends on both $\bm{x}$ and $\bm{m}$, 
and the anistropy may origin from either of them. 
The phases in which $f$ is independent of $\bm{x}$ are referred to as 
homogeneous phases. These phases show anistropy while keeping mobility in 
all directions. 
A typical example is the uniaxial nematic phase, where there exists a unit 
vector $\bm{n}$ such that $f=f((\bm{m}\cdot\bm{n})^2)$. 

Symmetry is always a central topic where anistropy appears. 
In liquid crystals, we need to discuss the symmetry at both macroscopic level 
and microscopic level: the phase symmetry and the molecular symmetry. 
The physical properties are mainly connected to the phase symmetry. 
Aiming at designing materials of physical properties more delicate, 
people have been striving for phases of other symmetries. 
This can be done by exerting external forces or confinements, but it brings 
limitation to application. 
With the hope of obtaining different phase symmetries spontaneously, 
people choose to alter the molecular symmetry. 
Among these molecules bent-core molecules have attracted considerable interests, 
whose rigid part possesses a bending (see the molecule in the middle of Fig. \ref{molecules}). 
Numerous unconventional liquid crystalline phases have been found for these molecules. 

Despite the rich phase behaviors obtained, by far no homogeneous phases 
has been found breaking the molecular symmetry. 
The uniaxial nematic phase, the only homogeneous phase rod-like molecules 
exhibit, is axisymmetric, identical to the symmetry of a rod. 
It is also the case for bent-core molecules, of which the homogeneous phases 
observed are restricted to the uniaxial and the biaxial nematic phases. 
Hence we would like to ask a question: will the molecular symmetry always be 
transmitted to phases? 
For further discussion, we need a clear mathematical formulation about phase and symmetry. 

\subsection{Rod-like molecules}
The theoretical study of liquid crystals begins from Onsager \cite{Ons}. 
He proposed a free energy functional for rods, 
\begin{align}
  F[f]=\int \d\bm{m} f(\bm{m})\log f(\bm{m})
  +\frac{c}{2}\int\d\bm{m}\d\bm{m'} 
    f(\bm{m})G(\bm{m},\bm{m'})f(\bm{m})\label{FreeEngN_rod}, 
\end{align}
where $c>0$ is an intensity parameter, and $f$ shall meet the normalization condition, 
\begin{equation}\label{Consrv_rod}
\int\d\nu f(\bm{m})=1. 
\end{equation}
The energy functional considers homogeneous phases only, 
as it does not include $\bm{x}$. 
Each phase corresponds to a local minimum. 

The energy functional is characterized by the kernel function $G$ 
that reflects the pairwise molecular interaction. 
Onsager considered the hard repulsive interaction and 
calculated the leading term of the excluded volume of two rods 
\begin{equation}
cG=2cl^2D|\bm{m}\times\bm{m'}| \label{Onsager}
\end{equation}
as the kernel function, where $l$ is the length and $D$ is the thickness. 
Later Maier and Saupe \cite{M_S} proposed a quadratic approximate kernel function, 
\begin{equation}
cG=c_2(\bm{m}\cdot\bm{m'})^2. \label{Maier_Saupe0}
\end{equation}
Both kernel are applied in the discussion of the isotopic -- uniaxial nematic phase transitions of rods. 
Because the polynomial form brings conveniences, the Maier-Saupe kernel has 
received much more attention, and is adopted widely in dynamic models \cite{Doi_book,Dyn1,Dyn2}. 

Axisymmetry is an important concept for rods. A rod is invariant when rotating 
it about its axis. This is why we can use the vector $\bm{m}$ to represent its 
oreientation. On the other hand, a phase is axisymmetric if $f$ is, 
which is expressed as 
$$
f=f(\bm{m}\cdot\bm{n}). 
$$
For the Maier-Saupe kernel, the axisymmetry of $f$ has been proved \cite{AxiSymMS,AxiSym2,AxiSym3}: 

\emph{The critical points of (\ref{FreeEngN_rod}) with the Maier-Saupe kernel (\ref{Maier_Saupe0}) shall satisfy $f=f((\bm{m}\cdot\bm{n})^2)$, where $\bm{n}$ is a unit vector.} 

Armed with this result, it is not difficult to find all the solutions. 
It also provides a solid foundation for the well-known Oseen-Frank theory 
\cite{Oseen_Frank} and Ericksen-Leslie theory \cite{Ericksen_Leslie}, 
which are built based on the axisymmetric assumption much earlier. 

\subsection{General formulation}
Although an elegant result has been acquired for rods, things become much more 
complicated for generic rigid molecules. 
When dealing with these molecules, we need a right-handed body-fixed 
orthonormal frame $(\bm{m}_1(P),\bm{m}_2(P),\bm{m}_3(P))$ to represent the 
orientation of a molecule. 
The variable $P\in SO(3)$ determines the orientation of the frame. 
The matrix representation of $P$ can be written as 
\begin{equation}
  P=\left(
  \begin{array}{ccc}
    m_{11} & m_{21} & m_{31} \\
    m_{12} & m_{22} & m_{32} \\
    m_{13} & m_{23} & m_{33} 
  \end{array}
  \right), 
\end{equation}
where $m_{ij}(P)=\bm{m}_i\cdot\bm{e}_j$ denotes the $j$th component of 
$\bm{m}_i$ in the space-fixed right-handed orthonormal frame 
$(\bm{e}_1,\bm{e}_2,\bm{e}_3)$. 
They can be expressed with three Euler angles
$$
\alpha\in [0,\pi],\ \beta,\gamma\in [0,2\pi)
$$
by 
\begin{align}
&P(\alpha,\beta,\gamma)\nonumber\\
=&\left(
\begin{array}{ccc}
 \cos\alpha &\quad -\sin\alpha\cos\gamma &\quad\sin\alpha\sin\gamma\\
 \sin\alpha\cos\beta &\quad\cos\alpha\cos\beta\cos\gamma-\sin\beta\sin\gamma &
 \quad -\cos\alpha\cos\beta\sin\gamma-\sin\beta\cos\gamma\\
 \sin\alpha\sin\beta &\quad\cos\alpha\sin\beta\cos\gamma+\cos\beta\sin\gamma &
 \quad -\cos\alpha\sin\beta\sin\gamma+\cos\beta\cos\gamma
\end{array}
\right).\label{EulerRep}
\end{align}
In this case, the uniform probability measure on $SO(3)$ is given by
$$
\d\nu=\frac{1}{8\pi^2}\sin\alpha\d\alpha\d\beta\d\gamma. 
$$
Sometimes we also use $P$ to represent the body-fixed frame. 

The energy functional is now written as
\begin{align}
  F[f]=\int \d\nu f(P)\log f(P)
  +\frac{c}{2}\int\d\nu(P)\d\nu(P') 
    f(P)G(P,P')f(P')\label{FreeEngN}, 
\end{align}
with the normalization condition 
\begin{equation}\label{Consrv}
\int\d\nu f(P)=1. 
\end{equation}
The kernel function 
$G$ depends only on the relative orientation $\bar{P}=P^{-1}P'$, 
whose elements are denoted by 
$$
p_{ij}=\bm{m}_i\cdot\bm{m'}_j. 
$$

We have proved in \cite{SymmO} that the kernel function inherits the molecular 
symmetry. Hence we would like to explain first how the molecular symmetry is 
expressed mathematically. 
All the orthogonal transformations that leave a molecule invariant form a 
three-dimensional point group $\H$. If the molecule is achiral, 
it can be divided into proper and improper rotations $\H=\H^+\cup H^-$; 
if the molecule is chiral, then $\H=\H^+\subseteq SO(3)$. 
When acting a transformation $T\in \H$ on a molecule, 
its body-fixed frame $P$ is converted into another frame $PT$ (see Fig. \ref{frames}). 
If $T\in \H^+$, the new frame is also right-handed; 
and if $T\in\H^-$, the new frame is left-handed. 
The set $\H^+$ is a subgroup of $\H$, and $\H^-$ is its coset: 
$$
\H^-=\H^+J=J\H^+,\quad \forall J\in\H^-. 
$$
The inheritance of molecular symmetry is expressed as \cite{SymmO}
\begin{equation}
G(\bar{P})=G(T\bar{P}T'),\qquad T,T'\in \H^+\ \mbox{or}\ T,T'\in\H^-, \label{Gsymm}
\end{equation}
which leads to $f(PT)=f(P)$ for $T\in\H^+$. 
The equation of $f$ holds naturally as $PT$ and $P$ substantially represent the same orientation. 
\begin{figure}
\centering
  \includegraphics[width=.5\textwidth]{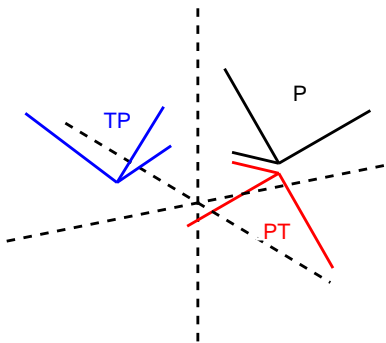}
  \caption{Rotation of an orthogonal frame. The dashed lines represent the frame $(\bm{e}_i)$. When rotated about $P$ itself, we obtain $PT$; when rotated about $(\bm{e}_i)$, we obtain $TP$. }
  \label{frames}
\end{figure}

We considered four different molecular symmetries: 
$\H=D_{\infty h},C_{\infty v},D_{2h},C_{2v}$. 
Here we use the Sch\"onflies notation: $C_n$ and $D_n$ represent the cyclic and 
dihedral group with $n$-fold rotation, respectively; $v$ and $h$ indicate a 
mirror plane parallel and vertical to the rotational axis, respectively. 
Some molecules of these symmetries are drawn in Fig. \ref{molecules}. 
Each of them is generated by inflating all the points in a set $A$ to a sphere 
of the same diameter $D$. For a rod, $A$ is a line segment; for a bent-core 
molecule, $A$ is a broken line with two equal segments. And for the other two 
molecules, we add the prefix 'sphero' to the shape of $A$: for an isosceles 
spherotriangle, $A$ is an isosceles triangle (including the interior and the boundary); for a spherocuboid, $A$ is a cuboid. 
These molecules are regarded fully rigid. 
The body-fixed orthonormal frame for each molecule 
is posed as drawn in Fig. \ref{molecules}, where $\bm{m}_1$ is always the 
rotational axis. 
Rods are of the $D_{\infty h}$ symmetry ($C_{\infty v}$ if with polarity). 
They possess axisymmetry about $\bm{m}_1$ and a mirror plane parallel to $\bm{m}_1$. 
If without polarity, they also have two-fold rotational symmetry about 
any direction vertical to $\bm{m}_1$ and a mirror plane vertical to $\bm{m}_1$. 
Bent-core molecules and isosceles spherotriangles are of the $C_{2v}$ symmetry. 
They possess two-fold rotational symmetry about $\bm{m}_1$ and a mirror plane 
parallel to $\bm{m}_1$. 
Spherocuboids are of the $D_{2h}$ symmetry. They possess two-fold rotational 
symmetries about $\bm{m}_i$, and mirror planes vertical to $\bm{m}_i$. 
Provided that the rotational axis is identical, the four point groups satisfy 
$$
D_{\infty h}\subseteq C_{\infty v}\subseteq D_{2h}\subseteq C_{2v}. 
$$
One could easily perceive the above relation by comparing the molecules in Fig. \ref{molecules}. 

The equality (\ref{Gsymm}) of $G$ determines its form if we require $G$ to be a quadratic 
polynomial of $p_{ij}$. It is also discussed in \cite{SymmO}. 
Let $\bm{m}_1$ coincide with the rotational axis. Then we have 
\begin{itemize}
\item $\H=D_{\infty h}$ generates 
\begin{equation}
cG=c_2(\bm{m}_1\cdot\bm{m'}_1)^2=c_2p^2_{11}. \label{Maier_Saupe}
\end{equation}
Let $\bm{m}=\bm{m}_1$. We have proved in \cite{SymmO} that the configuration space can be reduced to $S^2$. 
In this way we recover the Maier-Saupe kernel. 
\item $\H=C_{\infty v}$ generates 
\begin{equation}
  cG=c_1p_{11}+c_2p_{11}^2 \label{PolRod}, 
\end{equation}
which is used to examine the rods with polar magnetism \cite{Dipol}. 
\item $\H=D_{2h}$ generates 
\begin{equation}
  cG(\bar{P})=c_2p_{11}^2+c_3p_{22}^2+c_4(p_{12}^2+p_{21}^2), \label{BiKer}
\end{equation}
which is introduced by Starley \cite{Bi1} in a form linearly equivalent. 
The kernel function later received extensive numerical study by Virga et al
\cite{pre67,pre71,pre71_2,pre72,pre73}. 
Some dynamic models also use this kernel \cite{pre78,jr2009,ra2010,cms2010}. 
\item $\H=C_{2v}$ generates 
\begin{equation}\label{QuadApp}
  cG(\bar{P})=c_1p_{11}+c_2p_{11}^2+c_3p_{22}^2+c_4(p_{12}^2+p_{21}^2), 
\end{equation}
which is proposed in \cite{SymmO} and suitable for bent-core molecules. 
\end{itemize}
As these kernel functions are deduced from the molecular symmetry, 
we use the symmetry to name them. 
For instance, we name (\ref{PolRod}) the $C_{\infty v}$ kernel. 
Another thing that should be noted is that the $C_{2v}$ kernel (\ref{QuadApp}) 
can cover the other three kernels, for we may set some coefficients to zero. 
\begin{figure}
\centering
\includegraphics[width=\textwidth, keepaspectratio]{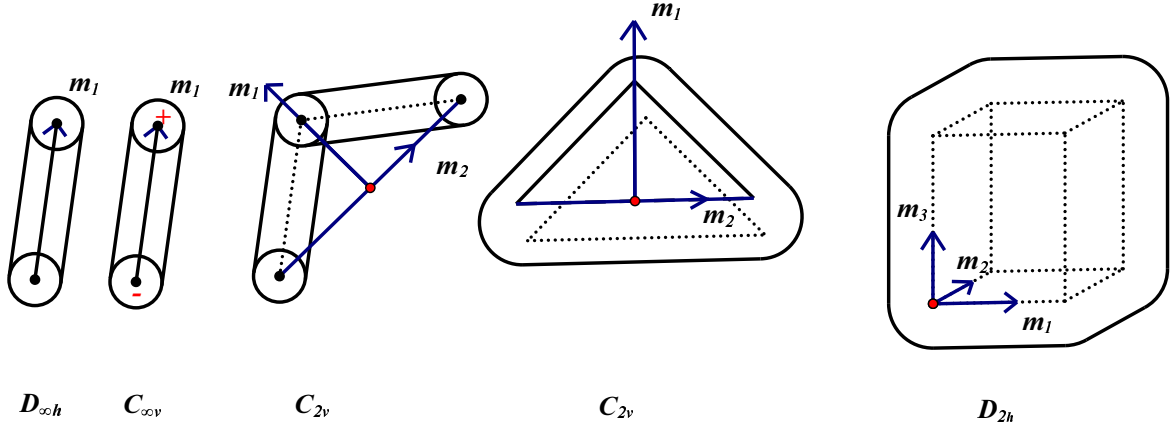}
\caption{Rigid molecules of different symmetries. From left to right: rod; rod with polarity; bent-core molecule; isosceles spherotriangle; spherocuboid.}\label{molecules}
\end{figure}

The symmetry of a phase, however, is expressed differently. 
When observing a phase, we are actually measuring some quantities in 
a space-fixed orthonormal frame. We say that a phase is symmetric under $T$, 
if the quantities are invariant when we rotate (possibly along with a 
reflection) all the molecules about the space-fixed frame with $T$. 
As these quantities are averages about $f$, 
we need to require that $f$ is invariant under this transformation. 
Note that the frame $P$ is transformed into $TP$ (see Fig. \ref{frames}). 
When $T$ is improper, without changing the orientation of a molecule, 
we may recover the new frame $TP$ to a right-handed one $TPJ$ with $J\in \H^-$. 
Therefore for any $P\in SO(3)$, it requires that 
\begin{align*}
&f(TP)=f(P), \quad\mbox{if }|T|=1, \\
&f(TPJ)=f(P),\ \forall J\in \H^-, \quad\mbox{if }|T|=-1.
\end{align*}
Denote by $\J_f$ the point group formed by such $T$. 
It depends on the choice of the space-fixed frame $\bm{e}_i$: 
when we rotate the frame $(\bm{e}_i)$ with $R^T\in SO(3)$, $\J_f$ becomes $R\J_f R^T$. 

To require the transmission of molecular symmetry to phases, 
it is necessary that $\H\subseteq R\J_f R^T$ for an $R$. 
The result for the Maier-Saupe kernel can be restated as
\begin{itemize}
\item For all local minima $f$ of the Maier-Saupe kernel (\ref{Maier_Saupe}), 
$D_{\infty h}\subseteq R\J_fR^T$ for an $R$. 
\end{itemize}
We would like to claim the following conjecture: 
\begin{itemize}
\item For all local minima $f$ of the $C_{\infty v}$ kernel (\ref{PolRod}), $C_{\infty v}\subseteq R\J_fR^T$ for an $R$. 
\item For all local minima $f$ of the $D_{2h}$ kernel (\ref{BiKer}), $D_{2h}\subseteq R\J_fR^T$ for an $R$. 
\item For all local minima $f$ of the $C_{2v}$ kernel (\ref{QuadApp}), $C_{2v}\subseteq R\J_fR^T$ for an $R$. 
\end{itemize}
The conjecture is supported by existing analytical and numerical results. 
For the $C_{\infty v}$ kernel, some relevant results are proved and the numerical results also suggest it \cite{Dipol}: although critical points are found $C_{\infty v}\nsubseteq R\J_fR^T$ for any $R$, all of them turn out to be unstable. 
For the $C_{2v}$ kernel, our earlier simulation in \cite{SymmO} of 
a special case suggests that $C_{2v}\subseteq D_{2h}\subseteq R\J_fR^T$ for an $R$. 
The relevant analytical results will be stated later in detail. 

In the current paper, we will prove the following result. 
\begin{theorem}\label{diag1}
For all local minima $f$ of the $C_{2v}$ kernel with $c_1\ge -1$, 
we have $D_{2h}\subseteq R\J_fR^T$ for an $R\in SO(3)$ if either of the following condition holds: 
\begin{enumerate}[(a)]
\item The quadratic form $c_2x^2+2c_4xy+c_3y^2$ is not negative definite; 
\item It is negative definite, but 
$$
\frac{c_4^2}{c_3}-c_2\le 2. 
$$
\end{enumerate}
\end{theorem}

The rest of the paper is organized as follows. In Sec. \ref{cond}, we derive the 
equivalent conditions of $\H\subseteq R\J_fR^T$ for the four symmetries. 
In Sec. \ref{proof}, we give the proof and application of the theorem. 
A concluding remark is given in Sec. \ref{concl}. 

\section{The equivalent condition\label{cond}}

Before continuing our discussion on the phase symmetry, 
we write down the critical points of the energy functional. 
Generally, the Euler-Lagrange equation of (\ref{FreeEngN}) yields
\begin{equation}\label{Boltz}
  f(P)=\frac{1}{Z}\exp\big(-W(P)\big), 
\end{equation}
where
\begin{equation}\label{DefU}
  W(P)=c\int\d\nu(P')G(\bar{P})f(P'), 
\end{equation}
and 
\begin{equation}
  Z=\int\d\nu(P)\exp\big(-W(P)\big).
\end{equation}
If the kernel function is a polynomial of $p_{ij}$, 
the Euler-Lagrange equation can be reduced to a few equations of tensors. 
With the kernel (\ref{QuadApp}), we can write the energy functional as
$$
F=\int \d\nu f\log f+c_1|\bm{p}|^2+c_2|Q_2|^2+c_3|Q_2|^2+2c_4Q_1:Q_2, 
$$
where $\bm{p}$, $Q_1$ and $Q_2$ are angular moments, 
$$
\bm{p}=\left<\bm{m}_1\right>,\quad Q_1=\left<\bm{m}_1\bm{m}_1\right>, \quad
Q_2=\left<\bm{m}_2\bm{m}_2\right>, 
$$
and their components are denoted by $p_i$ and $Q_{\alpha,ij}$ for $i,j=1,2,3$. 
Here we use the notation $\left<u\right>=\int\d\nu~u(P)f(P)$, and dots for 
tensor contraction. 
And $W(P)$ can be written as 
\begin{align}
W(P)=&c_{1}\bm{p}\cdot\bm{m}_1
+\big(c_{2}Q_1+c_{4}Q_2\big):\bm{m}_1\bm{m}_1
+\big(c_{3}Q_2+c_{4}Q_1\big):\bm{m}_2\bm{m}_2.\label{weight}
\end{align}
The tensors shall satisfy the following equations, 
\begin{align}
  \bm{p}&=\frac{1}{Z}\int\d\nu(P')\bm{m'}_1\exp\big(-W(P')\big), \label{SCp}\\
  Q_1&=\frac{1}{Z}\int\d\nu(P')\bm{m'}_1\bm{m'}_1\exp\big(-W(P')\big), \label{SCQ1}\\
  Q_2&=\frac{1}{Z}\int\d\nu(P')\bm{m'}_2\bm{m'}_2\exp\big(-W(P')\big). \label{SCQ2}
\end{align}

We write down the elements of $\H=\H^+\cup\H^-$ for the four point groups. 
Here we suppose that the rotational axis coincides with $\bm{m}_1$. 
\begin{align}
D_{\infty h}&=D_{\infty}\cup D_{\infty}J_3=\left\{\left(
\begin{array}{ccc}
  \pm 1 & 0 & 0\\
  0 & \pm\cos\theta & -\sin\theta\\
  0 & \pm\sin\theta & \cos\theta
\end{array}
\right)\right\}\bigcup
\left\{\left(
\begin{array}{ccc}
  \pm 1 & 0 & 0\\
  0 & \pm\cos\theta & \sin\theta\\
  0 & \pm\sin\theta & -\cos\theta
\end{array}
\right)\right\}, \label{D_infh}\\
C_{\infty v}&=C_{\infty}\cup C_{\infty}J_3=\left\{\left(
\begin{array}{ccc}
  1 & 0 & 0\\
  0 & \cos\theta & -\sin\theta\\
  0 & \sin\theta & \cos\theta
\end{array}
\right)\right\}\bigcup
\left\{\left(
\begin{array}{ccc}
  1 & 0 & 0\\
  0 & \cos\theta & \sin\theta\\
  0 & \sin\theta & -\cos\theta
\end{array}
\right)\right\}, \label{C_infv}\\
D_{2h}&=\{I,R_1,R_2,R_3\}\cup\{J_3,J_3R_1,J_3R_2,J_3R_3\}, \label{D_2h}\\
C_{2v}&=\{I,R_1\}\cup\{J_3,J_3R_1\}. \label{C_2v}
\end{align}
In the above, 
$$
R_1=\mbox{diag}(1,-1,-1),\quad R_2=\mbox{diag}(-1,1,-1),\quad R_3=\mbox{diag}(-1,-1,1),\quad J_3=\mbox{diag}(1,1,-1). 
$$
It is easy to verify for the four symmetries that if $W$ is written in 
(\ref{weight}), then it already holds that $f(PT)=f(P)$ for $T\in \H^+$. 
For example, for the $C_{\infty v}$ kernel, we have 
$$
W(P)=c_1\bm{p}\cdot\bm{m}_1+c_2Q_1:\bm{m}_1\bm{m}_1. 
$$
From $\bm{m}_1(PT)=\bm{m}_1(P)$ for $T\in C_{\infty}$, we deduce that $W(PT)=W(P)$. 

Next we discuss the condition for $\H\subseteq R\J_fR^T$. 
\begin{lemma}\label{condition}
Let the four point groups be written in (\ref{D_infh})-(\ref{C_2v}). 
  \begin{enumerate}[a)]
  \item For the $D_{\infty h}$ kernel (\ref{Maier_Saupe}), $D_{\infty h}\subseteq R\J_fR^T$ if and only if $RQ_1R^T$ is diagonal with $(RQ_1R^T)_{22}=(R^TQ_1R)_{33}$. 
  \item For the $C_{\infty v}$ kernel (\ref{PolRod}), $C_{\infty v}\subseteq R\J_fR^T$ if and only if $RQ_1R^T$ is diagonal with $(RQ_1R^T)_{22}=(RQ_1R^T)_{33}$, and 
$(R\bm{p})_2=(R\bm{p})_3=0$. 
  \item For the $D_{2h}$ kernel (\ref{BiKer}), $D_{2h}\subseteq R\J_fR^T$ if and only if both $RQ_1R^T$ and $RQ_2R^T$ are diagonal. 
  \item For the $C_{2v}$ kernel (\ref{QuadApp}), $C_{2v}\subseteq R\J_fR^T$ if and only if $RQ_1R^T$, $RQ_2R^T$ are diagonal and $(R\bm{p})_2=(R\bm{p})_3=0$. 
  \end{enumerate}
\end{lemma}
\begin{proof}
\begin{enumerate}[a)]
\item
For the Maier-Saupe kernel, let $T\in D_{\infty h}$. 
Taking $W(R^TTRP)=W(P)$ into $W(P)=c_2Q_1:\bm{m}_1\bm{m}_1$, we deduce that 
$$
(R^TT^TR)Q_1(R^TTR)=Q_1, \quad T\in D_{\infty}. 
$$
Or equivalently, 
$$
T^T(RQ_1R^T)T=RQ_1R^T, \quad T\in D_{\infty}. 
$$
Hence $RQ_1R^T$ shall be diagonal with $(RQ_1R^T)_{22}=(RQ_1R^T)_{33}$. 
\item
For the $C_{\infty v}$ kernel, we have $W(P)=c_1\bm{p}\cdot\bm{m}_1+c_2Q_1:\bm{m}_1\bm{m}_1$, leading to 
$$
T^T(R\bm{p})=R\bm{p},\ T^T(RQ_1R^T)T=RQ_1R^T, \quad T\in C_{\infty}. 
$$
In this case, $RQ_1R^T$ shall still be diagonal with $(RQ_1R^T)_{22}=(RQ_1R^T)_{33}$, 
and only the first component of $R\bm{p}$ can be nonzero. 
\item
For the $D_{2h}$ kernel, we have 
$W(P)=\big(c_{2}Q_1+c_{4}Q_2\big):\bm{m}_1\bm{m}_1
+\big(c_{3}Q_2+c_{4}Q_1\big):\bm{m}_2\bm{m}_2$. 
It can be deduced that 
$$T^T(RQ_1R^T)T=RQ_1R^T,\ T^T(RQ_2R^T)T=RQ_2R^T,\quad T\in\{R_1,R_2,R_3\}. $$
It follows that the off-diagonal elements of $R^TQ_iR$ equal to zero. 
\item
For the $C_{2v}$ kernel, we have 
$$T^T(R\bm{p})=R\bm{p},\ T^T(RQ_1R^T)T=RQ_1R^T,\ T^T(RQ_2R^T)T=RQ_2R^T,\quad T=R_1,J_3. $$
It requires that $(R\bm{p})_2=(R\bm{p})_3=0$ and that $RQ_iR^T$ are diagonal. 
\end{enumerate}
On the other hand, if the tensors meet the above conditions, 
it is easy to verify that $\H\subseteq R\J_fR^T$. 
\end{proof}

Recall that $R$ stands for our choice of the space-fixed frame $(\bm{e}_i)$. 
Hence it is sufficient that there exists a frame 
such that the tensors satisfy the conditions in the above Lemma. 
For example, that $RQ_1R^T$ and $RQ_2R^T$ are both diagonal 
means that we can choose a frame $(\bm{e}_i)$ in which both 
$Q_1$ and $Q_2$ are diagonalized. 
In the following lemma, we summarize the existing results 
in the language of the tensors. 
\begin{lemma}\label{ex}
Let $(\bm{p},Q_1,Q_2)$ be the solution of (\ref{SCp})-(\ref{SCQ2}). 
\begin{enumerate}[(i)]
\item For the Maier-Saupe kernel, two of the eigenvalues of $Q_1$ are equal \cite{AxiSymMS, AxiSym2, AxiSym3}. 
\item For the $C_{\infty v}$ and the $C_{2v}$ kernel, if $c_1\ge -1$, 
then $\bm{p}=0$ \cite{Dipol,SymmO}. 
\item For the $C_{\infty v}$ kernel, $\bm{p}$ is an eigenvector of $Q_1$ \cite{Dipol}. 
For the $C_{2v}$ kernel, if there exists a frame $(\bm{e}_i)$ in which both 
$Q_1$ and $Q_2$ are diagonalized, then $\bm{p}$ is an eigenvector of $Q_i$ \cite{SymmO}. 
\end{enumerate}
\end{lemma}
We compare Lemma \ref{ex} with Lemma \ref{condition}. 
For the $D_{\infty h}$ kernel, it is completely proved; 
for the $C_{\infty v}$ kernel, it is still an open problem that if $\bm{p}\ne 0$, 
then $Q_{1}$ has two equal eigenvalues in the subspace vertical to $\bm{p}$; 
for the other two kernels, 
we need to prove that there exists a frame $(\bm{e}_i)$ in which 
both $Q_1$ and $Q_2$ are diagonalized. 

\section{Proof and application\label{proof}}
In fact, we have proposed in \cite{SymmO} a very special condition of the 
coefficients such that there exists a frame $(\bm{e}_i)$ in which 
both $Q_1$ and $Q_2$ are diagonalized. 
But the condition is too strong. 
In Theorem \ref{diag1}, we extend the condition so that it can be applied to 
some molecules. 
\renewcommand{\proofname}{Proof of Theorem 1}
\begin{proof}
We know that $\bm{p}=0$ from Lemma \ref{ex}. Therefore
$$
W(P)=\big(c_{2}Q_1+c_{4}Q_2\big):\bm{m}_1\bm{m}_1
+\big(c_{3}Q_2+c_{4}Q_1\big):\bm{m}_2\bm{m}_2.
$$

\begin{enumerate}[(a)]
\item
  Write the quadratic form in the standard form, 
  $$ 
  c_2x^2+c_3y^2+2c_4xy=\lambda_1(d_1x+d_2y)^2+\lambda_2(d_2x-d_1y)^2. 
  $$
  We may suppose that $\lambda_2\ge 0$. 
  Hence
  $$
  W(P)=\lambda_1(d_1Q_1+d_2Q_2):(d_1\bm{m}_1\bm{m}_1+d_2\bm{m}_2\bm{m}_2)
  +\lambda_2(d_2Q_1-d_1Q_2):(d_2\bm{m}_1\bm{m}_1-d_1\bm{m}_2\bm{m}_2). 
  $$
  Denote 
  $$
  \tilde{Q}_1=d_1Q_1+d_2Q_2,\quad \tilde{Q}_2=d_2Q_1-d_1Q_2. 
  $$
  and 
  $$
  \tilde{q}_1=d_1\bm{m}_1\bm{m}_1+d_2\bm{m}_2\bm{m}_2,
  \quad \tilde{q}_2=d_2\bm{m}_1\bm{m}_1-d_1\bm{m}_2\bm{m}_2. 
  $$
  Select a space-fixed frame such that $\tilde{Q}_1$ is diagonal. 
  We will show that $\tilde{Q}_2$ is also diagonal in this frame. 
  Let $J_1=\mbox{diag}(-1,1,1)$, $J_3=\mbox{diag}(1,1,-1)$. then 
  $$
  m_{i1}(J_1PJ_3)=-m_{i1}(P),m_{i2}(J_1PJ_3)=m_{i2}(P),m_{i3}(J_1PJ_3)=m_{i3}(P),\quad i=1,2. 
  $$
  Let
  $$
  W_1(P)=\lambda_1\tilde{Q}_1:\tilde{q}_1
  +\lambda_2(\tilde{Q}_{2,ii}\tilde{q}_{2,ii}+2\tilde{Q}_{2,23}\tilde{q}_{2,23}). 
  $$
  We have $W_1(J_1PJ_3)=W_1(P)$. Therefore
  \begin{align*}
    &\tilde{Q}_{2,12}^2+\tilde{Q}_{2,13}^2\\
    =&\frac{\int \d\nu \exp(-W_1(P))
      (\tilde{Q}_{2,12}\tilde{q}_{2,12}+\tilde{Q}_{2,13}\tilde{q}_{2,13})
      \sinh(-2\lambda_2(Q_{1,12}\tilde{q}_{2,12}+Q_{1,13}\tilde{q}_{2,13}))}
    {\int \d\nu \exp(-W_1(P))
      \cosh(-2\lambda_2(Q_{1,12}\tilde{q}_{2,12}+Q_{1,13}\tilde{q}_{2,13}))}. 
  \end{align*}
  Since $\lambda_2\ge 0$, the right side $\le 0$. 
  Similarly, we can prove that $\tilde{Q}_{2,23}=0$. 
  Thus $\tilde{Q}_2$ is diagonal. 
\item 
  From the condition, we can find $d_1,d_2$ and $0<\epsilon\le 2$ such that
  $$-c_2=\epsilon+d_1^2,~-c_3=d_2^2,~-c_4=d_1d_2. $$ 
  Hence
  $$
  W(P)=-(d_1Q_1+d_2Q_2):(d_1\bm{m}_1\bm{m}_1+d_2\bm{m}_2\bm{m}_2)
  -\epsilon Q_1:\bm{m}_1\bm{m}_1. 
  $$
  Similar to the first part of the theorem, we may suppose that $d_1Q_1+d_2Q_2$ is diagonal, and let
  $$
  W_1(P)=-(d_1Q_1+d_2Q_2):(d_1\bm{m}_1\bm{m}_1+d_2\bm{m}_2\bm{m}_2)
  -\epsilon(Q_{1,ii}m_{1i}^2+2Q_{1,23}m_{12}m_{13}). 
  $$
  It also holds $W_1(J_1PJ_3)=W_1(P)$, giving
  \begin{align*}
    &Q_{1,12}^2+Q_{1,13}^2\\
    =&\frac{\int \d\nu \exp(-W_1(P))
      (Q_{1,12}m_{11}m_{12}+Q_{1,13}m_{11}m_{13})
      \sinh(2\epsilon(Q_{1,12}m_{11}m_{12}+Q_{1,13}m_{11}m_{13}))}
    {\int \d\nu \exp(-W_1(P))
      \cosh(2\epsilon(Q_{1,12}m_{11}m_{12}+Q_{1,13}m_{11}m_{13}))}. 
  \end{align*}
  Since $0<\epsilon\le 2$, using $x\tanh(x)<x^2~(x\ne 0)$, we obtain
  \begin{align*}
    &Q_{1,12}^2+Q_{1,13}^2\\
    \le &\frac{\int \d\nu \exp(-W_1(P))
      2\epsilon(Q_{1,12}m_{11}m_{12}+Q_{1,13}m_{11}m_{13})^2
      \cosh(2\epsilon(Q_{1,12}m_{11}m_{12}+Q_{1,13}m_{11}m_{13}))}
    {\int \d\nu \exp(-W_1(P))
      \cosh(2\epsilon(Q_{1,12}m_{11}m_{12}+Q_{1,13}m_{11}m_{13}))}.
  \end{align*}
  But 
  \begin{align*}
  (Q_{1,12}m_{11}m_{12}+Q_{1,13}m_{11}m_{13})^2
  \le &m_{11}^2(Q_{1,12}^2+Q_{1,13}^2)(m_{12}^2+m_{13}^2)\\
  =&(Q_{1,12}^2+Q_{1,13}^2)m_{11}^2(1-m_{11}^2)\\
  \le &\frac{1}{4}(Q_{1,12}^2+Q_{1,13}^2). 
  \end{align*}
  Therefore we get
  $$
  Q_{1,12}^2+Q_{1,13}^2\le\frac{\epsilon}{2}(Q_{1,12}^2+Q_{1,13}^2)
  \le Q_{1,12}^2+Q_{1,13}^2, 
  $$
  leading to $Q_{1,12}=Q_{1,13}=0$. Thus $Q_1$ is diagonal. 
\end{enumerate}
\end{proof}

\begin{figure}
  \centering
  \includegraphics[width=.25\textwidth]{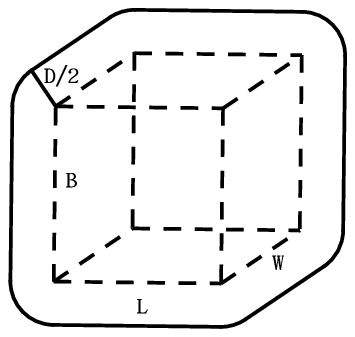}
  \includegraphics[width=.3\textwidth]{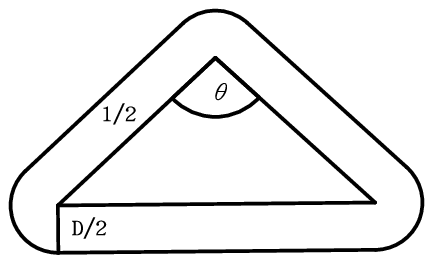}
  \includegraphics[width=.3\textwidth]{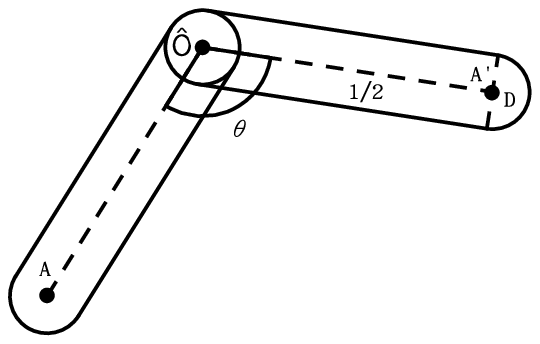}
  \caption{Molecular parameters}\label{pars}
\end{figure}
Now we apply the theorem to the molecules drawn in Fig. \ref{molecules}. 
The coefficients in the kernel function can be written as functions of molecular parameters.  
This is done for all the four molecules by approximating the excluded volume 
using various methods. 
The parameters include (see Fig. \ref{pars}): 
the diameter of sphere $D$; 
for isosceles spherotriangles and bent-core molecules, 
the length of lateral or arm $l/2$, the top angle $\theta$; 
and for spherocuboids, the length of three edges $W,B,L$. 

For cuboids (the case $D=0$), the coefficients given by Starley \cite{Bi1}, 
interpolated from the excluded volume at specific orientations, are
\begin{align*}
  c_2 &=c\left[-B(W^2+L^2)-W(L^2+B^2)+4WBL-(L^2-BW)(B-W)\right],\\
  c_3 &=c\left[-B(W^2+L^2)-W(L^2+B^2)+4WBL+(L^2-BW)(B-W)\right], \\
  c_4 &=c\left[-B(W^2+L^2)-W(L^2+B^2)+L(W^2+B^2)+2WBL\right]. 
\end{align*}
The coefficients given in \cite{quadproj} for spherocuboids, 
based on the projection of excluded volume in \cite{Mulder}, are
\begin{align*}
  c_2 &=\frac{15c}{16}\left[-B(W^2+L^2)-W(L^2+B^2)+4WBL-(L^2-BW)(B-W)-\frac{\pi D}{2}(L-B)^2\right],\\
  c_3 &=\frac{15c}{16}\left[-B(W^2+L^2)-W(L^2+B^2)+4WBL+(L^2-BW)(B-W)-\frac{\pi D}{2}(L-W)^2\right], \\
  c_4 &=\frac{15c}{16}\left[-B(W^2+L^2)-W(L^2+B^2)+L(W^2+B^2)+2WBL-\frac{\pi D}{2}(L-W)(L-B)\right]. 
\end{align*}
When $D=0$, they are proportional to the Starley's. 
The coefficients for spherotriangles, computed in \cite{SymmO} also by projection, are 
\begin{align*}
  c_{1}&=\frac{3}{8}cl^2DK(\theta)\ge 0, \\
  c_{2}&=-\frac{15}{64}cl^3\sin\theta\cos^2\frac{\theta}{2}
  -\frac{15\pi}{128}cl^2D\cos^4\frac{\theta}{2},\\
  c_{3}&=-\frac{15}{64}cl^3\sin\theta\sin\frac{\theta}{2}(1+\sin\frac{\theta}{2})
  -\frac{15\pi}{128}cl^2D\sin^2\frac{\theta}{2}(1+\sin\frac{\theta}{2})^2, \\
  c_{4}&=-\frac{15}{128}cl^3\sin\theta(1+\sin\frac{\theta}{2})
  -\frac{15\pi}{128}cl^2D\cos^2\frac{\theta}{2}\sin\frac{\theta}{2}
  (1+\sin\frac{\theta}{2}). 
\end{align*}
For bent-core molecules, the coefficients can be calculated numerically as is 
described in \cite{SymmO}. They are proportional to $cl^3$ and depend on two 
dimensionless parameters $D/l$ and $\theta$. 
It needs to be pointed out that $c_1\ge 0$ for bent-core molecules. 

\begin{figure}
\centering
\includegraphics[width=.49\textwidth, keepaspectratio]{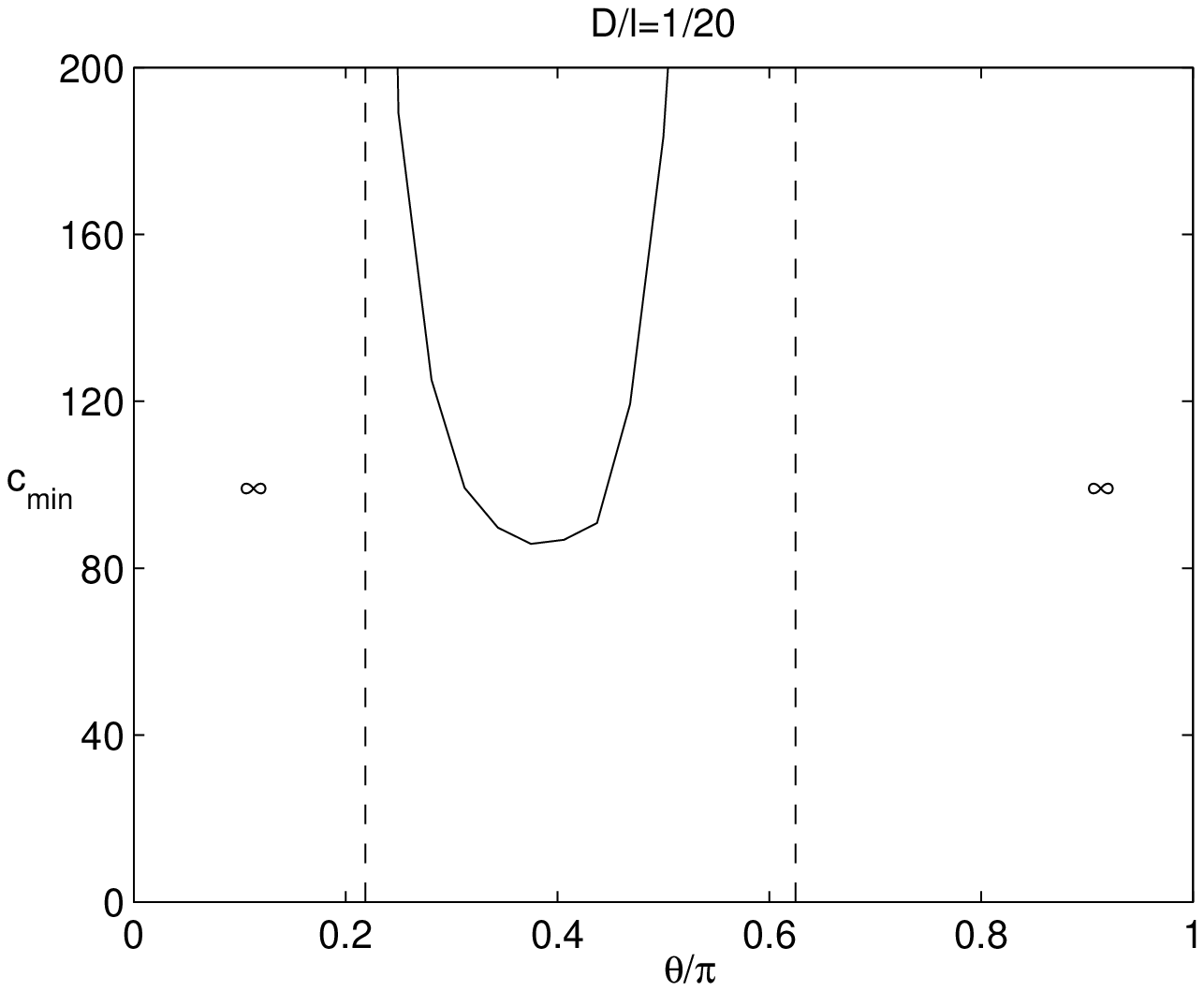}
\includegraphics[width=.49\textwidth, keepaspectratio]{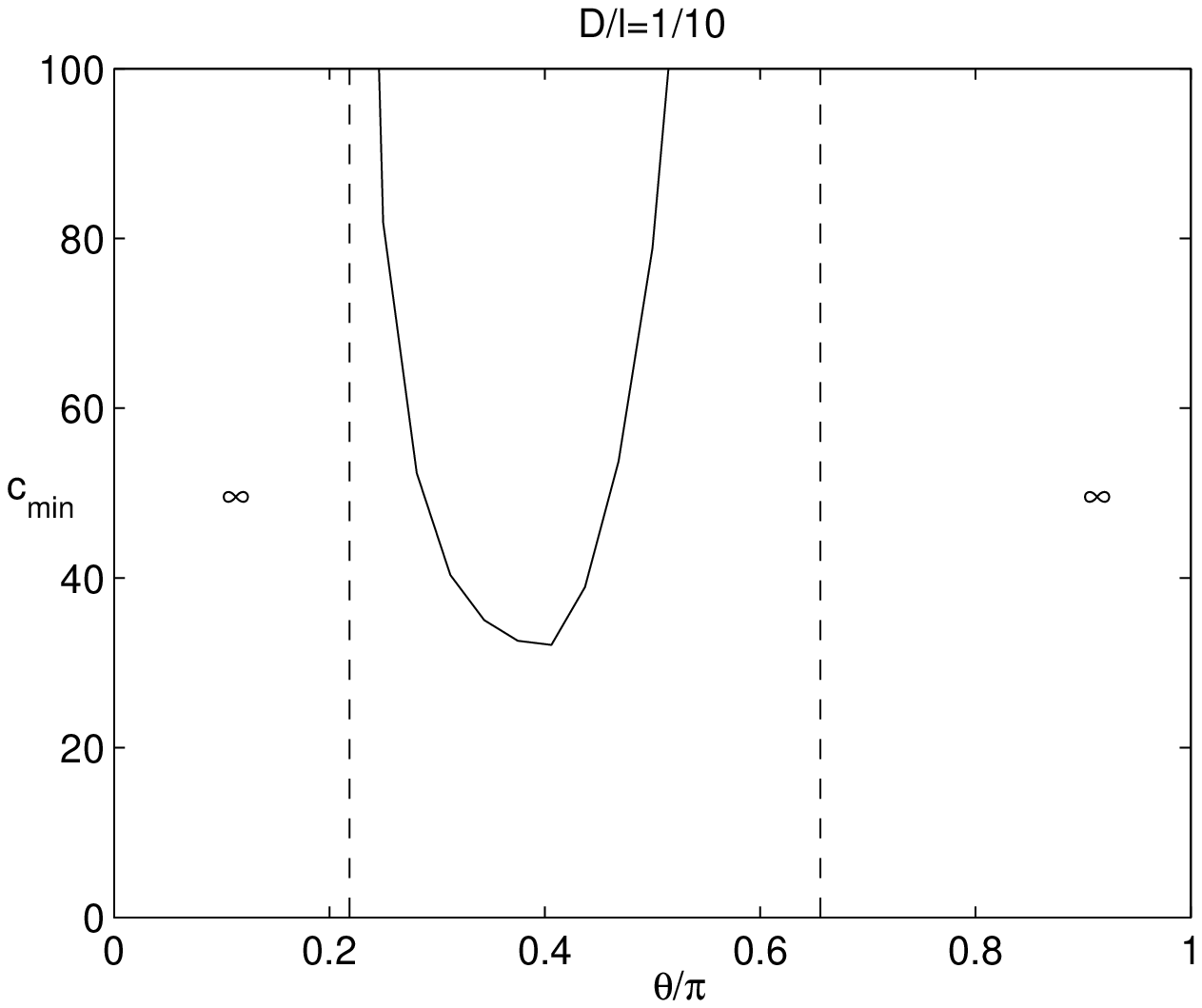}
\includegraphics[width=.49\textwidth, keepaspectratio]{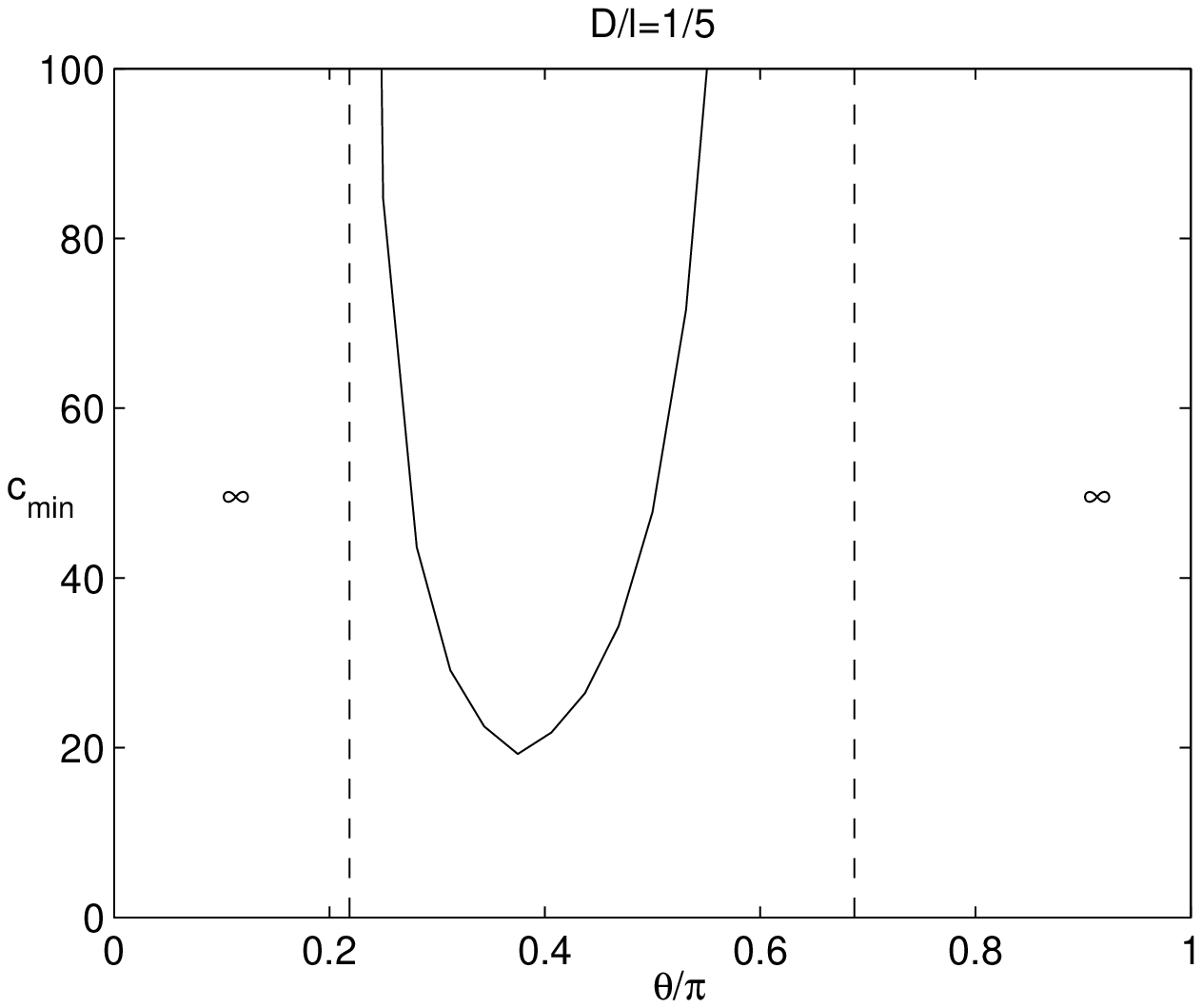}
\caption{Minimal $c$ that makes $\epsilon>2$, $c_{\mbox{\tiny min}}$. }\label{cmin}
\end{figure}
The coefficients derived from spherocuboids and spherotriangles satisfy the first condition. 
In fact, for spherocuboids, it gives
$$
c_4^2-c_2c_3=k(W-B)^2(B-L)^2(L-W)^2\ge 0, 
$$
where $k$ is a positive number. 
For spherotriangles, we have 
$$
c_4^2-c_2c_3=\left(\frac{15cl^3\sin\theta}{128}\right)^2
(2\sin\frac{\theta}{2} - 1)^2(\sin\frac{\theta}{2} + 1)^2\ge 0. 
$$
Since the product of two eigenvalues equals to $c_2c_3-c_4^2\le 0$, 
it follows from the theorem that for both molecules, 
$Q_1$ and $Q_2$ share an eigenframe. 

Now we turn to bent-core molecules. 
Because $c_i$ are propotional to the $cl^3$, we set $l=2$ without loss of generality. Then we have 
$$
\epsilon=-\frac{c_2c_3-c_4^2}{c_3}=-c\cdot\frac{c_2^{(0)}c_3^{(0)}-(c_4^{(0)})^2}{c_3^{(0)}}, 
$$
if $c_2c_3-c_4^2>0$, where $c_i^{(0)}$ stands for the value at $c=1$. 
Note that the theorem still holds when switching $c_2$ and $c_3$. 
Therefore the minimal $c$ to make $\epsilon\ge 2$ is
$$
c_{\mbox{\tiny min}}=\frac{\max\{-2c_2^{(0)},-2c_3^{(0)}\}}{c_2^{(0)}c_3^{(0)}-(c_4^{(0)})^2}. 
$$
We calculate $c_{\mbox{\tiny min}}$ for $D/l=1/20,1/10,1/5$, plotted in Fig. \ref{cmin}. 
In the regions outside the dashed line, which are labeled with $\infty$, 
it holds $c_4^2-c_2c_3\ge 0$. 
In the intermediate region, the value of $c_{\mbox{\tiny min}}$ is large enough to generate modulation, which is discussed in another paper \cite{tensor}. 

\section{Concluding remarks\label{concl}}
We have proved that $Q_1$ and $Q_2$ share an eigenframe conditionally. 
Here we would like to provide more computational results suggesting that it 
holds always for the kernel (\ref{QuadApp}). 
In fact, we do simulation with $c_4=0$ and $c_2+c_3=-20\ (c_2,c_3\le 0),\ c_1\in[0,3]$. 
Even if $c_1=0$, it is far from the condition in the theorem. 
To evaluate the distance between two eigenframes, we calculate the Frobenius 
norm $||Q_1Q_2-Q_2Q_1||_F$, which equals to zero when two eigenframes coincide. 
It turns out that $||Q_1Q_2-Q_2Q_1||_F\le 10^{-9}$, 
indicating that $Q_1$ and $Q_2$ shall share an eigenframe. 

Summarizing the existing results, we claim a conjecture that the phase symmetry 
maintains molecular symmetry for the quadratic kernels determined by 
the $D_{\infty h}$, $C_{\infty v}$, $D_{2h}$ and $C_{2v}$ symmetries. 
We give a proof with a condition that is applicable to three classes of molecules. 
A complete proof is yet to be reached and shall be an interesting problem. 
It is also intriguing to see whether it holds for higher-order kernel and other symmetries. 

\vspace{10pt}
{\textbf{Acknowledgement}} Dr. Yucheng Hu provides some useful suggestions. 
P. Zhang is partly supported by National Natural Science Foundation of China 
(Grant No. 11421101 and No. 11421110001).

\bibliographystyle{plain}
\bibliography{bib_diag} 

\end{document}